\documentclass[11pt,reqno]{amsart}

\usepackage{amsmath, amsthm, amstext, amssymb, amsfonts, graphicx, bm, subfigure}
\newcommand{\PR}{\mathbb P}
\newcommand{\E}{\mathbb E}
\newcommand{\C}{\mathbb C}
\newcommand{\mytilde}{\raise.17ex\hbox{$\scriptstyle\mathtt{\sim}$}}

\newtheorem{theorem}{Theorem}[section]
\newtheorem{corollary}[theorem]{Corollary}
\newtheorem{definition}[theorem]{Definition}

\newtheorem{proposition}[theorem]{Proposition}
\newtheorem*{assumption}{Assumptions}
\newtheorem*{remark}{Remark}

\begin{document}

\title{Critical phenomena in exponential random graphs}

\author{Mei Yin}
\address{Department of Mathematics, University of Texas, Austin,
TX, 78712, USA} \email{myin@math.utexas.edu}

\dedicatory{\rm September 19, 2013}

\begin{abstract}
The exponential family of random graphs is one of the most
promising class of network models. Dependence between the random
edges is defined through certain finite subgraphs, analogous to
the use of potential energy to provide dependence between particle
states in a grand canonical ensemble of statistical physics. By
adjusting the specific values of these subgraph densities, one can
analyze the influence of various local features on the global
structure of the network. Loosely put, a phase transition occurs
when a singularity arises in the limiting free energy density, as
it is the generating function for the limiting expectations of all
thermodynamic observables. We derive the full phase diagram for a
large family of 3-parameter exponential random graph models with
attraction and show that they all consist of a first order surface
phase transition bordered by a second order critical curve.
\end{abstract}

\maketitle

\section{Introduction}
The exponential family of random graphs is one of the most widely
studied network models. Their popularity lies in the fact that
they capture a wide variety of common network tendencies, such as
connectivity and reciprocity, by representing a complex global
structure through a set of tractable local features. The
theoretical foundations for these models were originally laid by
Besag \cite{Besag}, who applied methods of statistical analysis
and demonstrated the powerful Markov-Gibbs equivalence
(Hammersley-Clifford theorem \cite{HC}) in the context of spatial
data. Building on Besag's work, further investigations quickly
followed. Holland and Leinhardt \cite{HL} derived the exponential
family of distributions for networks in the directed case. Frank
and Strauss \cite{FS} showed that the random graph edges form a
Markov random field when the local network features are given by
counts of various triangles and stars. Newer developments are
summarized in Snijders et al. \cite{S} and Rinaldo et al.
\cite{R}. (See Wasserman and Faust \cite{WF} for a comprehensive
review of the methods and models for analyzing network
properties.)

As usual in statistical physics, we start with a finite
probability space, namely the set $\mathcal{G}_n$ of all simple
graphs on $n$ vertices (``simple'' means undirected, with no loops
or multiple edges). The general $k$-parameter family of
exponential random graphs is defined by assigning a probability
mass function $\PR_n^{\beta}(G_n)$ to every simple graph $G_n \in
\mathcal{G}_n$:
\begin{equation}
\label{pmf} \PR_n^{\beta}(G_n)=\exp\left(n^2(\beta_1
t(H_1,G_n)+\cdots+
  \beta_k t(H_k,G_n)-\psi_n^{\beta})\right),
\end{equation}
where $\beta=(\beta_1,...,\beta_k)$ are $k$ real parameters,
$H_1,...,H_k$ are pre-chosen finite simple graphs (in particular,
we take $H_1$ to be a single edge), $t(H_i, G_n)$ is the density
of graph homomorphisms (the probability that a random vertex map
$V(H_i) \to V(G_n)$ is edge-preserving),
\begin{equation}
t(H_i, G_n)=\frac{|\text{hom}(H_i, G_n)|}{|V(G_n)|^{|V(H_i)|}},
\end{equation}
and $\psi_n^{\beta}$ is the normalization constant (free energy
density),
\begin{equation}
\label{psi} \sum_{G_n \in \mathcal{G}_n} \exp\left(n^2
\left(\beta_1 t(H_1,G_n)+\cdots+\beta_k t(H_k,G_n)\right)
\right)=\exp\left(n^2 \psi^\beta_n\right).
\end{equation}

These exponential random graphs are particularly useful when one
wants to simulate observed networks as closely as possible, but
without going into details of the specific process underlying
network formation. Since real-world networks are often very large
in size, ranging from hundreds to billions of vertices, our main
interest will be in the behavior of the exponential random graph
$G_n$ in the large $n$ limit. Intuitively, the $k$ parameters
$\beta_1,...,\beta_k$ allow one to adjust the influence of
different local features (in this case, densities of different
subgraphs $H_1,...,H_k$) on the limiting probability distribution,
and a natural question to ask is how would the tuning of
parameters impact the global structure of a typical random graph
$G_n$ drawn from this model? Even in the dense graph regime where
the number of edges in the graph scales like $O(n^2)$, this
question is already interesting, and so this paper focuses on
large dense random graphs with non-negative parameters $\beta_i$.
Realistic networks are often fairly sparse. Nevertheless, if the
parameters $\beta_i$ in the model are sufficiently large negative
(i.e., high concentrations of certain local features are
discouraged), then typical realizations of the exponential model
would exhibit sparse behavior, and limiting graph structures in
this region will be addressed in a forthcoming paper.

Loosely put, a phase transition occurs when the limiting free
energy density $\displaystyle \psi^\beta_\infty=\lim_{n\to
  \infty}\psi_n^{\beta}$ has a singular point. The
reason behind this is that the limiting free energy density is the
generating function for the limiting expectations of all
thermodynamic observables,
\begin{equation}
\label{E} \lim_{n\to \infty}\E^\beta t(H_i, G_n)=\lim_{n\to
\infty}\frac{\partial}{\partial
\beta_i}\psi_n^\beta=\frac{\partial}{\partial
\beta_i}\psi_\infty^\beta,
\end{equation}
\begin{equation}
\label{Cov} \lim_{n\to \infty}n^2\left(\C\textrm{ov}^\beta
\left(t(H_i, G_n), t(H_j, G_n)\right)\right)=\lim_{n\to
\infty}\frac{\partial^2}{\partial \beta_i
\partial \beta_j}\psi_n^\beta=\frac{\partial^2}{\partial \beta_i
\partial \beta_j}\psi_\infty^\beta.
\end{equation}
Notice that the exchange of limits in (\ref{E}) and (\ref{Cov}) is
nontrivial, since it involves summation over an infinite number of
terms. Building on earlier work of Chatterjee and Diaconis
\cite{CD}, we will show in Theorem \ref{One} that $\displaystyle
\psi_{\infty}^{\beta}$ exists and explore its analyticity
properties. The proof of Theorem $2$ by Yang and Lee \cite{YL} on
the commutation of limits then goes through without much
difficulty in this setting, as the free energy density under
consideration here may also be expressed as (locally) uniformly
convergent power series. This implies that a singularity in the
limiting thermodynamic function must arise from a singularity in
the limiting free energy density, and we can define phases and
phase transitions through the limiting free energy density as
follows.

\begin{definition}
A phase is a connected region of the parameter space $\{\beta\}$,
maximal for the condition that the limiting free energy density
$\displaystyle \psi_{\infty}^{\beta}$ is analytic. There is a
$j$th-order transition at a boundary point of a phase if at least
one $j$th-order partial derivative of $\psi_{\infty}^{\beta}$ is
discontinuous there, while all lower order derivatives are
continuous.
\end{definition}

For $k=1$, it has been well established that the exponential model
reduces to the famous Erd\H{o}s-R\'{e}nyi random graph $G(n,
\rho)$ \cite{ER}, which has on average $\binom{n}{2} \rho$ edges,
and its structure is completely specified by the edge formation
probability $\displaystyle \rho=e^{2\beta_1}/(1+e^{2\beta_1})$.
Fix a finite $n$. As $\rho$ increases, the model evolves from a
low-density state in which all components are small to a
high-density state in which an extensive fraction of all vertices
are joined together in a single giant component. In the large $n$
limit, the transition occurs when $\rho$ is close to $0$ or
equivalently when $\beta_1$ is close to $-\infty$. This phenomenon
coincides with our above definition, as in one dimension, the
limiting free energy density $\displaystyle \psi_{\infty}^{\beta}$
of the random graphs is analytic.

For $k=2$, the situation is understandably more complicated and
has attracted enormous attention in recent years: Park and Newman
\cite{PN1} \cite{PN2} developed mean-field approximations and
analyzed the phase diagram for the edge-$2$-star and edge-triangle
models. Chatterjee and Diaconis \cite{CD} gave the first rigorous
proof of singular behavior in the edge-triangle model with the
help of the emerging tools of graph limits as developed by
Lov\'{a}sz and coworkers \cite{LS}. There are also related results
in H\"{a}ggstr\"{o}m and Jonasson \cite{HJ} and Bhamidi et al.
\cite{B}. Radin and Yin \cite{RY} derived the full phase diagram
for $2$-parameter exponential random graph models with attraction
($\beta_2\geq 0$) and showed that they all contain a first order
transition curve ending in a second order critical point. Aristoff
and Radin \cite{AR} treated $2$-parameter random graph models with
repulsion ($\beta_2\leq 0$) and proved that the region of
parameter space corresponding to multipartite structure is
separated by a phase transition from the region of disordered
graphs (their proof was recently improved by Yin \cite{Yin}).

One of the key motivations for considering exponential random
graphs is to develop models that exhibit transitivity and clumping
(i.e., a friend of a friend is likely also a friend). However, as
seen in experiments and through heuristics \cite{PN2}, it is often
futile to model transitivity with only $2$ subgraphs $H_1$ and
$H_2$ (say edge and triangle) as sufficient statistics. If
$\beta_2$ is positive, the graph is essentially behaving like an
Erd\H{o}s-R\'{e}nyi graph, while if $\beta_2$ is negative, it
becomes roughly bipartite \cite{CD}. The near-degeneracy observed
in experiments and proved in \cite{CD} \cite{RY} for large values
of $\beta_2$ also renders the $2$-parameter model quite useless.
To accurately model the global structural properties of real-world
networks, more local features of the random graph $G_n$ need to be
captured. We therefore incorporate the density of one more
subgraph $H_3$ into the probability distribution and study the
phase structure of the exponential model in the $k=3$ setting. Our
main results are the following.

\begin{assumption}
Consider a $3$-parameter exponential random graph model where the
probability mass function $\PR_n^{\beta}(G_n)$ for $G_n \in
\mathcal{G}_n$ is given by
\begin{equation}
\label{assumption} \PR_n^{\beta}(G_n)=\exp\left(n^2(\beta_1
t(H_1,G_n)+\beta_2 t(H_2,G_n)+\beta_3
t(H_3,G_n)-\psi_n^{\beta})\right).
\end{equation}
Assume that $H_1$ is a single edge, $H_2$ has $p$ edges, and $H_3$
has $q$ edges, with $2 \leq p \leq q \leq 5p-1$.
\end{assumption}

\begin{theorem}
\label{One} Consider a $3$-parameter exponential random graph
model (\ref{assumption}). The limiting free energy density
$\displaystyle \psi_\infty^\beta$ exists at all $\{(\beta_1,
\beta_2, \beta_3): \beta_2\geq 0, \beta_3\geq 0\}$, and is
analytic except on a certain continuous surface $S$ which includes
three bounding curves $C_1$, $C_2$, and $C_3$: The surface $S$
approaches the plane $\beta_1+\beta_2+\beta_3=0$ as $\beta_1 \to
-\infty$, $\beta_2 \to \infty$, and $\beta_3 \to \infty$; The
curve $C_1$ is the intersection of $S$ with the $(\beta_1,
\beta_2)$ plane $\{(\beta_1, \beta_2, \beta_3): \beta_3=0\}$; The
curve $C_2$ is the intersection of $S$ with the $(\beta_1,
\beta_3)$ plane $\{(\beta_1, \beta_2, \beta_3): \beta_2=0\}$; The
curve $C_3$ is a critical curve, and is given parametrically by
\begin{eqnarray}
    \beta_1(u)&=&\frac{1}{2}\log
\frac{u}{1-u}-\frac{1}{2(p-1)(1-u)}+\frac{pu-(p-1)}{2(p-1)(q-1)(1-u)^2}, \notag\\
    \beta_2(u)&=&\frac{qu-(q-1)}{2p(p-1)(p-q)u^{p-1}(1-u)^2},  \notag\\
    \beta_3(u)&=&\frac{pu-(p-1)}{2q(q-1)(q-p)u^{q-1}(1-u)^2},
\end{eqnarray}
where we take $\displaystyle \frac{p-1}{p}\leq u\leq
\frac{q-1}{q}$ to meet the non-negativity constraints on $\beta_2$
and $\beta_3$ (see Figure \ref{critical}). All the first
derivatives $\displaystyle \frac{\partial}{\partial
\beta_1}\psi_\infty^\beta$, $\displaystyle
\frac{\partial}{\partial \beta_2}\psi_\infty^\beta$, and
$\displaystyle \frac{\partial}{\partial \beta_3}\psi_\infty^\beta$
have (jump) discontinuities across the surface $S$, except along
the curve $C_3$ where, however, all the second derivatives
$\displaystyle \frac{\partial^2}{\partial
\beta_1^2}\psi_\infty^\beta$, $\displaystyle
\frac{\partial^2}{\partial \beta_2^2}\psi_\infty^\beta$,
$\displaystyle \frac{\partial^2}{\partial
\beta_3^2}\psi_\infty^\beta$, $\displaystyle
\frac{\partial^2}{\partial \beta_1 \partial
\beta_2}\psi_\infty^\beta$, $\displaystyle
\frac{\partial^2}{\partial \beta_1 \partial
\beta_3}\psi_\infty^\beta$, and $\displaystyle
\frac{\partial^2}{\partial \beta_2 \partial
\beta_3}\psi_\infty^\beta$ diverge.
\end{theorem}

\begin{figure}
\centering
\includegraphics[clip=true,height=3.5in]{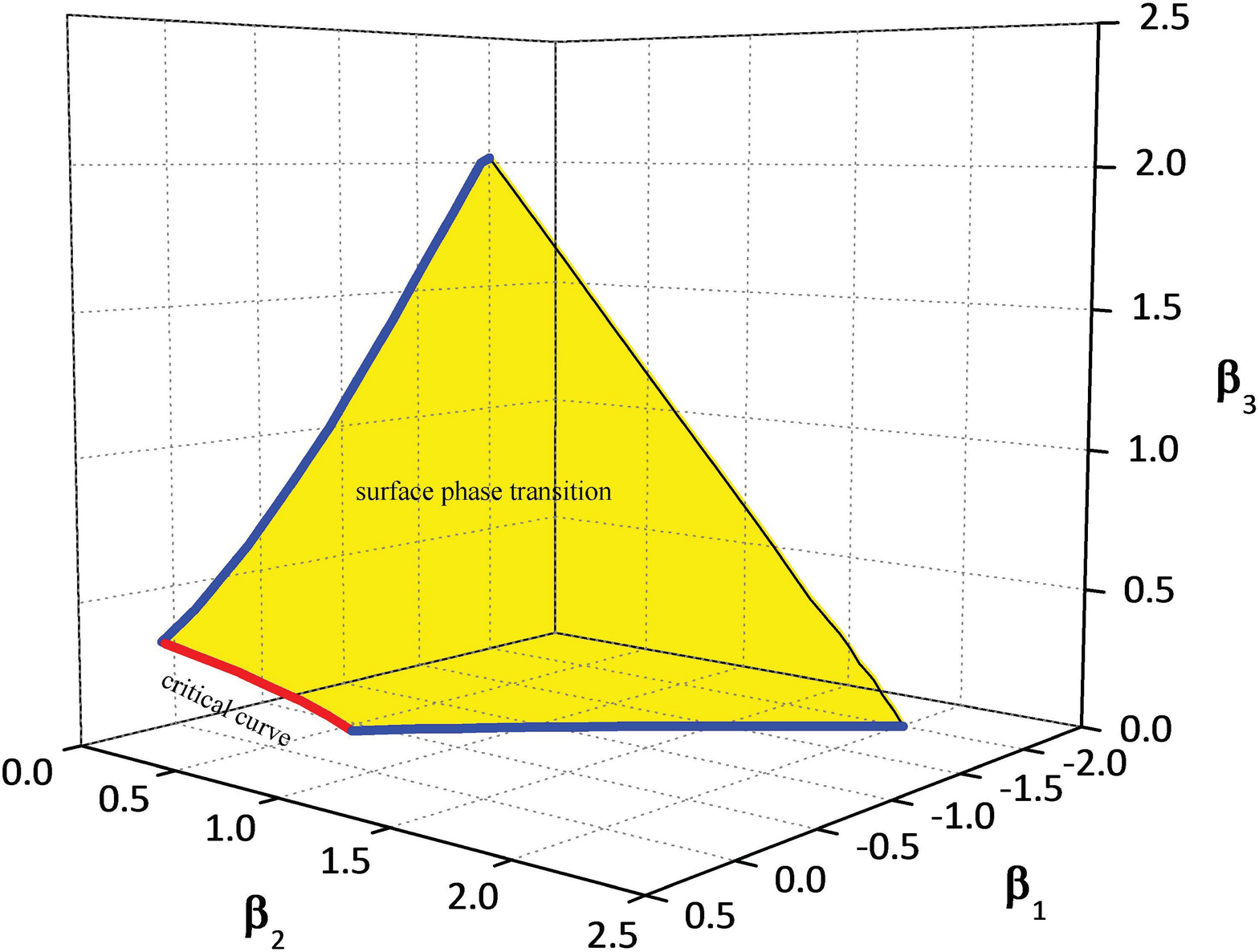}
\caption{Phase transition surface $S$ and critical curve $C_3$.
Graph drawn for $p=3$ and $q=5$.} \label{critical}
\end{figure}

By (\ref{E}) and (\ref{Cov}), the analyticity (or lack thereof) of
the limiting free energy density $\psi_\infty^\beta$ encodes
important information about the local features of the random graph
$G_n$ for large $n$: A (jump) discontinuity in the first
derivatives of $\psi_\infty^\beta$ across the surface $S$
indicates a discontinuity in the expected local densities, while
the divergence of the second derivatives of $\psi_\infty^\beta$
along the curve $C_3$ implies that the covariances of the local
densities go to zero more slowly than $1/n^2$.

\begin{corollary}
\label{Cor} The parameter space $\{(\beta_1, \beta_2, \beta_3):
\beta_2\geq 0, \beta_3\geq 0\}$ consists of a single phase with a
first order phase transition across the surface $S$ and a second
order phase transition along the critical curve $C_3$.
\end{corollary}

\begin{remark}
The requirement that the number of edges $p$ in $H_2$ and the
number of edges $q$ in $H_3$ satisfy $2 \leq p \leq q \leq 5p-1$
in the \textbf{Assumptions} is just a technicality. It is expected
that the parameter space would still consist of a single phase
with first order phase transition(s) across one (or more) surfaces
and second order phase transition(s) along the critical curves
should such assumptions fail.
\end{remark}

To derive these results, we will make use of two theorems from
\cite{CD}, which connect the occurrence of a phase transition in
our model with the solution of a certain maximization problem (a
more extensive explanation may be found in \cite{LS}).

\begin{theorem}[Theorem 4.1 in \cite{CD}]
\label{CD} Consider a general $k$-parameter exponential random
graph model (\ref{pmf}). Suppose $\beta_2,...,\beta_k$ are
non-negative. Then the limiting free energy density $\displaystyle
\psi_\infty^\beta$ exists, and is given by
\begin{equation}
\label{lmax} \psi_{\infty}^{\beta}=\sup_{0\leq u\leq
1}\left(\beta_1 u^{E(H_1)}+\cdots+\beta_k
u^{E(H_k)}-\frac{1}{2}u\log u-\frac{1}{2}(1-u)\log(1-u)\right)\\,
\end{equation}
where $E(H_i)$ is the number of edges in $H_i$.
\end{theorem}

\begin{theorem}[Theorem 4.2 in \cite{CD}]
\label{gen} Let $G_n$ be an exponential random graph drawn from
(\ref{pmf}). Suppose $\beta_2,...,\beta_k$ are non-negative. Then
$G_n$ behaves like an Erd\H{o}s-R\'{e}nyi graph $G(n, u^*)$ in the
large n limit, where $u^*$ is picked randomly from the set $U$ of
maximizers of (\ref{lmax}).
\end{theorem}

Given the Chatterjee-Diaconis result, computing phase boundaries
for the exponential model (\ref{assumption}) mainly reduces to a
$3$-dimensional calculus problem coupled with probability
estimates. However, as straight-forward as it sounds, to get a
clear picture of the limiting probability distribution and hence
the global structure of a typical random graph $G_n$ drawn from
this model, we need to solve the intricate calculus problem
explicitly and employ various tricks. This mechanism may be
generalized to a $k$-parameter setting (\ref{pmf}), and the
crucial idea (as will be illustrated in the proof of Proposition
\ref{max}) is to minimize the effect of the ordered parameters on
the limiting free energy density one by one.

The rest of this paper is organized as follows. In Section
\ref{MA} we analyze the maximization problem (\ref{lmax}) for
$k=3$ in detail (Proposition \ref{max}) and describe the
transition surface $S$ and the bounding curves $C_1$, $C_2$, and
$C_3$ explicitly (Proposition \ref{char}). In Section \ref{CB} we
investigate the analyticity properties of the limiting free energy
density $\psi^\beta_\infty$ in different parameter regions
(Theorems \ref{A} and \ref{B}) and complete the proof of our main
theorem (Theorem \ref{One}).

\section{Maximization Analysis}
\label{MA}
\begin{figure}
\centering
\includegraphics[clip=true, height=3.5in]{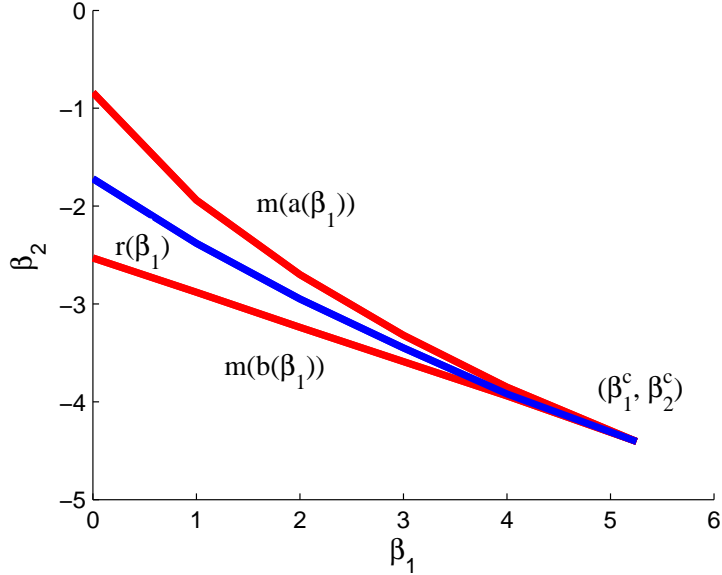}
\caption{The V-shaped region (with phase transition curve
$r(\beta_1)$ inside) in the $(\beta_1, \beta_2)$ plane. Graph
drawn for $\beta_3=2$, $p=3$, and $q=5$.} \label{Vshape}
\end{figure}

\begin{proposition}
\label{max} Fix $\beta_3$ and integers $p$ and $q$ with $2 \leq p
\leq q \leq 5p-1$. Consider the maximization problem for
\begin{equation}
\label{l} l_{\beta_3}(u; \beta_1,
\beta_2)=\beta_1u+\beta_2u^p+\beta_3u^q-\frac{1}{2}u\log
u-\frac{1}{2}(1-u)\log(1-u)
\end{equation}
on the interval $[0,1]$, where $-\infty<\beta_1<\infty$ and
$-\infty<\beta_2<\infty$ are parameters. Then there is a V-shaped
region in the $(\beta_1, \beta_2)$ plane with corner point
$(\beta_1^c, \beta_2^c)$,
\begin{eqnarray}
\label{point} \beta_1^c&=&\frac{1}{2}\log
\frac{u_0}{1-u_0}-\frac{1}{2(p-1)(1-u_0)}+\frac{pu_0-(p-1)}{2(p-1)(q-1)(1-u_0)^2}, \notag\\
\beta_2^c&=&\frac{qu_0-(q-1)}{2p(p-1)(p-q)u_0^{p-1}(1-u_0)^2},
\end{eqnarray}
where $u_0$ is uniquely determined by
\begin{eqnarray}
\label{beta3}
\beta_3=\frac{pu_0-(p-1)}{2q(q-1)(q-p)u_0^{q-1}(1-u_0)^2}.
\end{eqnarray}
Outside this region, $l_{\beta_3}(u)$ has only one local maximizer
(hence global maximizer) $u^*$; Inside this region,
$l_{\beta_3}(u)$ has exactly two local maximizers $u_1^*$ and
$u_2^*$. For every $\beta_1$ inside this V-shaped region
($\beta_1<\beta_1^c$), there is a unique decreasing
$\beta_2=r_{\beta_3}(\beta_1)$ such that $u_1^*$ and $u_2^*$ are
both global maximizers for $l_{\beta_3}(u; \beta_1,
r_{\beta_3}(\beta_1))$ (see Figures \ref{Vshape} and
\ref{curves}).
\end{proposition}

\begin{figure}
\centering
\includegraphics[clip=true, height=3.5in]{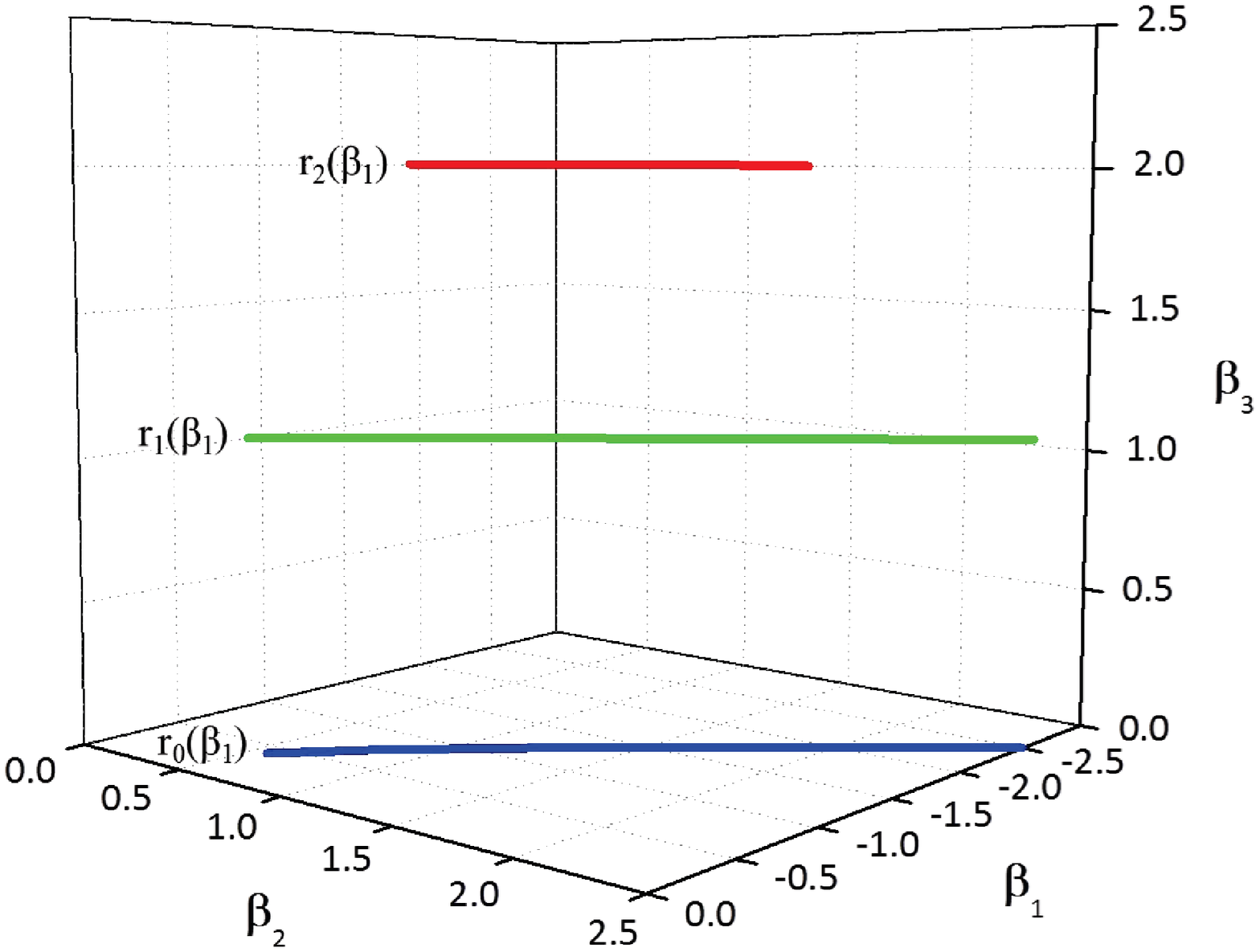}
\caption{The phase transition curves $r(\beta_1)$ (corresponding
to $\beta_3=0$, $\beta_3=1$, and $\beta_3=2$) in the $(\beta_1,
\beta_2, \beta_3)$ space. Graph drawn for $p=3$ and $q=5$.}
\label{curves}
\end{figure}

\begin{proof}
The location of maximizers of $l_{\beta_3}(u)$ on the interval
$[0,1]$ is closely related to the properties of its derivatives
$l'_{\beta_3}(u)$ and $l''_{\beta_3}(u)$:
\begin{equation}
\label{l'}
l'_{\beta_3}(u)=\beta_1+p\beta_2u^{p-1}+q\beta_3u^{q-1}-\frac{1}{2}\log\frac{u}{1-u},
\end{equation}
\begin{equation}
\label{l''}
l''_{\beta_3}(u)=p(p-1)\beta_2u^{p-2}+q(q-1)\beta_3u^{q-2}-\frac{1}{2u(1-u)}.
\end{equation}

We first analyze the properties of $l''_{\beta_3}(u)$ on the
interval $[0, 1]$. Consider instead
\begin{equation}
F(u)=p(p-1)\beta_2+q(q-1)\beta_3u^{q-p}-\frac{1}{2u^{p-1}(1-u)},
\end{equation}
which is obtained by factorizing $u^{p-2}$ out of
$l''_{\beta_3}(u)$. Note that in doing so the effect of $\beta_2$
is minimized as varying $\beta_2$ only shifts the graph of $F(u)$
upward/downward and does not affect its shape. To examine the
effect of $\beta_3$ on $F(u)$, we take one more derivative,
\begin{equation}
F'(u)=q(q-1)(q-p)\beta_3u^{q-p-1}+\frac{(p-1)-pu}{2u^{p}(1-u)^2}.
\end{equation}
Similarly as before, we factor $u^{q-p-1}$ out of $F'(u)$ to
minimize the effect of $\beta_3$. Let
\begin{equation}
f(u)=\frac{(p-1)-pu}{u^{q-1}(1-u)^2}
\end{equation}
so that
\begin{equation}
F'(u)=\frac{1}{2}u^{q-p-1}(2q(q-1)(q-p)\beta_3+f(u)).
\end{equation}

\begin{figure}
\centering
\includegraphics[clip=true, width=6in]{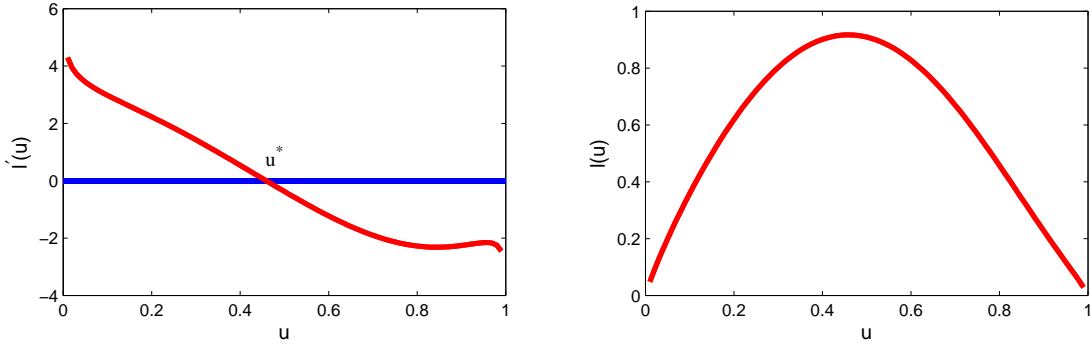}
\caption{Outside the V-shaped region, $l_{\beta_3}(u)$ has a
unique local maximizer (hence global maximizer) $u^*$. Graph drawn
for $\beta_1=2$, $\beta_2=-4$, $\beta_3=2$, $p=3$, and $q=5$.}
\label{below}
\end{figure}

We claim that the condition $2\leq p\leq q\leq 5p-1$ guarantees
that $f(u)$ is monotonically decreasing on $[0, 1]$. Independent
of $p$ and $q$, $f(0)=\infty$ and $f(1)=-\infty$. Its derivative
$f'(u)$ is given by
\begin{equation}
f'(u)=-\frac{pqu^2+(p+q+1-2pq)u+(p-1)(q-1)}{u^q(1-u)^3}.
\end{equation}
Rearranging terms in the discriminant $\Delta$ of the numerator of
$f'(u)$ yields a quadratic equation in $q$,
\begin{equation}
\Delta=q^2+2(1-3p)q+(p+1)^2
\end{equation}
with two zeros
\begin{equation}
q_{1,2}=(3p-1)\pm 2\sqrt{2(p^2-p)}.
\end{equation}
We can easily check that $q_1\leq p$ and $q_2\geq 5p-1$. As
$q_1\leq q\leq q_2$ is equivalent to $\Delta\leq 0$, this verifies
our claim.

An immediate corollary is that there is a unique $u_0$ in $(0, 1)$
such that $F'(u_0)=0$, with $F'(u)>0$ for $u<u_0$, and $F'(u)<0$
for $u>u_0$. The correspondence between $\beta_3$ and $u_0$ is
one-to-one, and we may alternatively describe $\beta_3$ by
\begin{equation}
\beta_3=\frac{pu_0-(p-1)}{2q(q-1)(q-p)u_0^{q-1}(1-u_0)^2}.
\end{equation}
This further implies that $F(u)$ is increasing from $0$ to $u_0$,
and decreasing from $u_0$ to $1$, with the global maximum achieved
at $u_0$,
\begin{equation}
F(u_0)=p(p-1)\beta_2+\frac{qu_0-(q-1)}{2(q-p)u_0^{p-1}(1-u_0)^2}.
\end{equation}
Let
\begin{equation}
\beta_2^c=\frac{qu_0-(q-1)}{2p(p-1)(p-q)u_0^{p-1}(1-u_0)^2}
\end{equation}
so that $F(u_0; \beta_2^c)=0$. As $F(u)$ and $l''_{\beta_3}(u)$
always carry the same sign, this shows that for $\beta_2 \leq
\beta_2^c$, $l''_{\beta_3}(u)\leq 0$ on the whole interval $[0,
1]$; whereas for $\beta_2>\beta_2^c$, $l''_{\beta_3}(u)$ takes on
both positive and negative values, and we denote the transition
points by $u_1$ and $u_2$ ($u_1<u_0<u_2$), which are solely
determined by $\beta_2$, and vice versa. Let
\begin{equation}
m(u)=\frac{1}{2p(p-1)u^{p-1}(1-u)}+\frac{(p-1)-pu_0}{2p(p-1)(q-p)u_0^{q-1}(1-u_0)^2}u^{q-p}
\end{equation}
so that $\beta_2=m(u_1)=m(u_2)$. As $m(u)=\beta_2-F(u)/p(p-1)$, we
have $m(0)=m(1)=\infty$, $m(u)$ is decreasing from $0$ to $u_0$,
and increasing from $u_0$ to $1$.

\begin{figure}
\centering
\includegraphics[clip=true, width=6in]{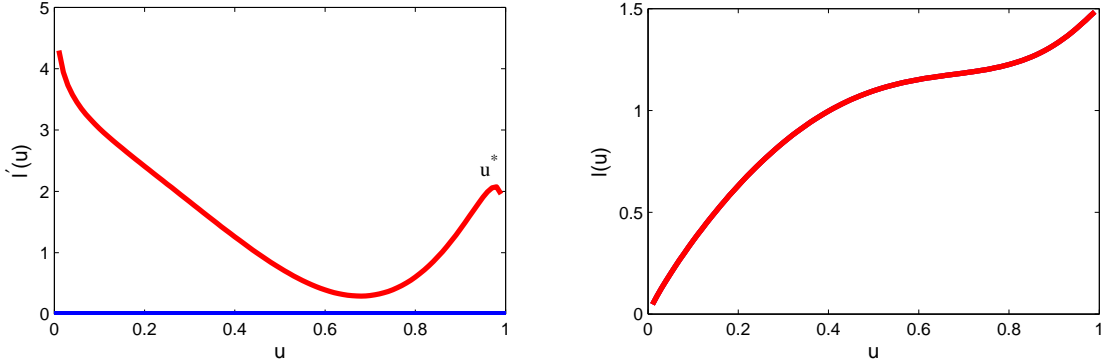}
\caption{Outside the V-shaped region, $l_{\beta_3}(u)$ has a
unique local maximizer (hence global maximizer) $u^*$. Graph drawn
for $\beta_1=2$, $\beta_2=-2.5$, $\beta_3=2$, $p=3$, and $q=5$.}
\label{above}
\end{figure}

Based on the properties of $l''_{\beta_3}(u)$, we next analyze the
properties of $l'_{\beta_3}(u)$ on the interval $[0, 1]$. For
$\beta_2 \leq \beta_2^c$, $l'_{\beta_3}(u)$ is monotonically
decreasing. For $\beta_2>\beta_2^c$, $l'_{\beta_3}(u)$ is
decreasing from $0$ to $u_1$, increasing from $u_1$ to $u_2$, then
decreasing again from $u_2$ to $1$. For reasons that will become
clear in a moment, we write down the explicit expressions of
$l'_{\beta_3}(u_1)$ and $l'_{\beta_3}(u_2)$:
\begin{equation}
l'_{\beta_3}(u_1)=\beta_1+\frac{1}{2(p-1)(1-u_1)}-\frac{1}{2}\log\frac{u_1}{1-u_1}+\frac{(p-1)-pu_0}{2(p-1)(q-1)u_0^{q-1}(1-u_0)^2}u_1^{q-1},
\end{equation}
\begin{equation}
l'_{\beta_3}(u_2)=\beta_1+\frac{1}{2(p-1)(1-u_2)}-\frac{1}{2}\log\frac{u_2}{1-u_2}+\frac{(p-1)-pu_0}{2(p-1)(q-1)u_0^{q-1}(1-u_0)^2}u_2^{q-1}.
\end{equation}

Finally, based on the properties of $l'_{\beta_3}(u)$ and
$l''_{\beta_3}(u)$, we analyze the properties of $l_{\beta_3}(u)$
on the interval $[0, 1]$. Independent of $p$ and $q$,
$l_{\beta_3}(u)$ is a bounded continuous function,
$l'_{\beta_3}(0)=\infty$, and $l'_{\beta_3}(1)=-\infty$, so
$l_{\beta_3}(u)$ can not be maximized at $0$ or $1$. For $\beta_2
\leq \beta_2^c$, $l'_{\beta_3}(u)$ crosses the $u$-axis only once,
going from positive to negative. Thus $l_{\beta_3}(u)$ has a
unique local maximizer (hence global maximizer) $u^*$. For
$\beta_2> \beta_2^c$, the situation is more complicated. If
$l'_{\beta_3}(u_1)\geq 0$ (resp. $l'_{\beta_3}(u_2)\leq 0$),
$l_{\beta_3}(u)$ has a unique local maximizer (hence global
maximizer) at a point $u^*>u_2$ (resp. $u^*<u_1$). If
$l'_{\beta_3}(u_1)<0<l'_{\beta_3}(u_2)$, then $l_{\beta_3}(u)$ has
two local maximizers $u_1^*$ and $u_2^*$, with
$u_1^*<u_1<u_0<u_2<u_2^*$ (see Figures \ref{below} and
\ref{above}).

\begin{figure}
\centering
\includegraphics[clip=true, width=6in]{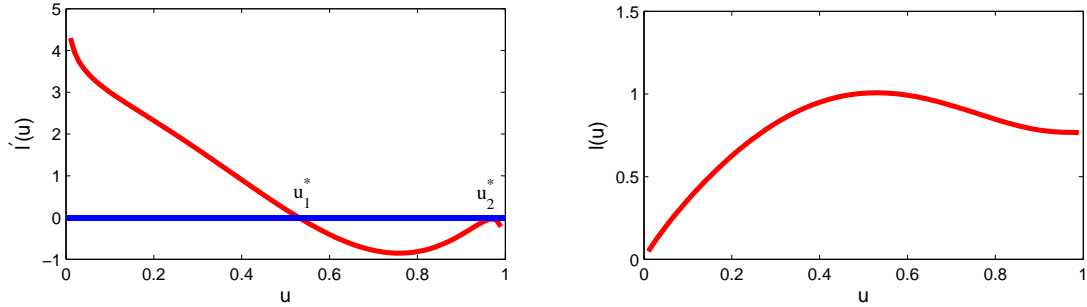}
\caption{Along the lower bounding curve $m(b(\beta_1))$ of the
V-shaped region, $l'_{\beta_3}(u)$ has two zeros $u_1^*$ and
$u_2^*$, but only $u_1^*$ is the global maximizer for
$l_{\beta_3}(u)$. Graph drawn for $\beta_1=2$, $\beta_2=-3.24$,
$\beta_3=2$, $p=3$, and $q=5$.} \label{lower}
\end{figure}

Let
\begin{equation}
n(u)=\frac{1}{2(p-1)(1-u)}-\frac{1}{2}\log\frac{u}{1-u}+\frac{(p-1)-pu_0}{2(p-1)(q-1)u_0^{q-1}(1-u_0)^2}u^{q-1}
\end{equation}
so that $l'_{\beta_3}(u_1)=\beta_1+n(u_1)$ and
$l'_{\beta_3}(u_2)=\beta_1+n(u_2)$. Independent of $p$ and $q$,
$n(0)=\infty$ and $n(1)=\infty$. Its derivative $n'(u)$ is given
by
\begin{eqnarray}
n'(u)&=&\frac{1}{2(p-1)}u^{q-2}\left(\frac{(p-1)-pu_0}{u_0^{q-1}(1-u_0)^2}-\frac{(p-1)-pu}{u^{q-1}(1-u)^2}\right)\notag\\
&=&\frac{1}{2(p-1)}u^{q-2}\left(f(u_0)-f(u)\right).
\end{eqnarray}
As $f(u)$ is monotonically decreasing, $n(u)$ is decreasing from
$0$ to $u_0$, and increasing from $u_0$ to $1$, with the global
minimum achieved at $u_0$,
\begin{equation}
n(u_0)=\frac{1}{2(p-1)(1-u_0)}-\frac{1}{2}\log
\frac{u_0}{1-u_0}+\frac{(p-1)-pu_0}{2(p-1)(q-1)(1-u_0)^2}.
\end{equation}
This implies that $l'_{\beta_3}(u_1; \beta_1, \beta_2^c) \geq0$
for
\begin{equation}
\beta_1\geq \beta_1^c=\frac{1}{2}\log
\frac{u_0}{1-u_0}-\frac{1}{2(p-1)(1-u_0)}+\frac{pu_0-(p-1)}{2(p-1)(q-1)(1-u_0)^2}.
\end{equation}
The only possible region in the $(\beta_1, \beta_2)$ plane where
$l'_{\beta_3}(u_1)<0<l'_{\beta_3}(u_2)$ is thus bounded by
$\beta_1<\beta_1^c$ and $\beta_2>\beta_2^c$.

\begin{figure}
\centering
\includegraphics[clip=true, width=6in]{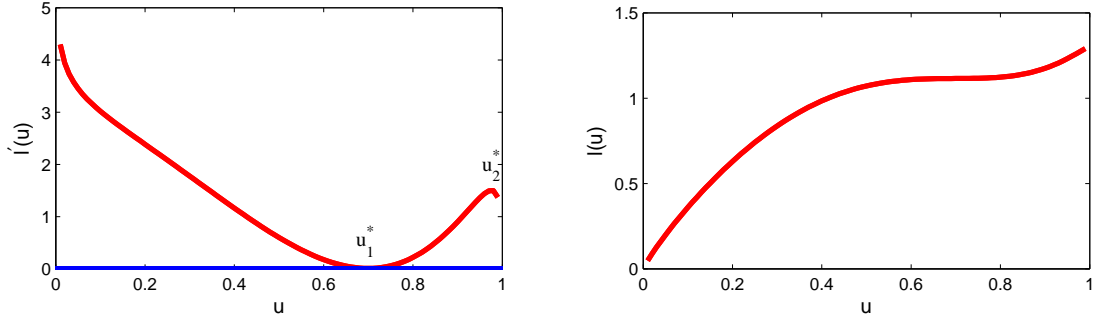}
\caption{Along the upper bounding curve $m(a(\beta_1))$ of the
V-shaped region, $l'_{\beta_3}(u)$ has two zeros $u_1^*$ and
$u_2^*$, but only $u_2^*$ is the global maximizer for
$l_{\beta_3}(u)$. Graph drawn for $\beta_1=2$, $\beta_2=-2.7$,
$\beta_3=2$, $p=3$, and $q=5$.} \label{upper}
\end{figure}

We now analyze the behavior of $l'_{\beta_3}(u_1)$ and
$l'_{\beta_3}(u_2)$ more closely when $\beta_1$ and $\beta_2$ are
chosen from this region. Recall that $u_1<u_0<u_2$. By
monotonicity of $n(u)$ on the intervals $(0, u_0)$ and $(u_0, 1)$,
there exist continuous functions $a(\beta_1)$ and $b(\beta_1)$ of
$\beta_1$, such that $l'_{\beta_3}(u_1)<0$ for $u_1>a(\beta_1)$
and $l'_{\beta_3}(u_2)>0$ for $u_2>b(\beta_1)$. As $\beta_1\to
-\infty$, $a(\beta_1)\to 0$ and $b(\beta_1)\to 1$. $a(\beta_1)$ is
an increasing function of $\beta_1$, whereas $b(\beta_1)$ is a
decreasing function, and they satisfy
\begin{equation}
n(a(\beta_1))=n(b(\beta_1))=-\beta_1.
\end{equation}
The restrictions on $u_1$ and $u_2$ yield restrictions on
$\beta_2$, and we have $l'_{\beta_3}(u_1)<0$ for
$\beta_2<m(a(\beta_1))$ and $l'_{\beta_3}(u_2)>0$ for
$\beta_2>m(b(\beta_1))$. As $\beta_1\to -\infty$,
$m(a(\beta_1))\to \infty$ and $m(b(\beta_1))\to \infty$.
$m(a(\beta_1))$ and $m(b(\beta_1))$ are both decreasing functions
of $\beta_1$, and they satisfy
\begin{equation}
l'_{\beta_3}(u_1; \beta_1, m(a(\beta_1)))=l'_{\beta_3}(u_2;
\beta_1, m(b(\beta_1)))=0.
\end{equation}
As $l'_{\beta_3}(u_2; \beta_1, \beta_2)>l'_{\beta_3}(u_1; \beta_1,
\beta_2)$ for every $(\beta_1, \beta_2)$, the curve
$m(b(\beta_1))$ must lie below the curve $m(a(\beta_1))$, and
together they generate the bounding curves of the V-shaped region
in the $(\beta_1, \beta_2)$ plane with corner point $(\beta_1^c,
\beta_2^c)$ where two local maximizers exist for $l_{\beta_3}(u)$
(see Figures \ref{lower} and \ref{upper}).

Fix an arbitrary $\beta_1<\beta_1^c$, we examine the effect of
varying $\beta_2$ on the graph of $l'_{\beta_3}(u)$. It is clear
that $l'_{\beta_3}(u)$ shifts upward as $\beta_2$ increases and
downward as $\beta_2$ decreases. As a result, as $\beta_2$ gets
large, the positive area bounded by the curve $l'_{\beta_3}(u)$
increases, whereas the negative area decreases. By the fundamental
theorem of calculus, the difference between the positive and
negative areas is the difference between $l_{\beta_3}(u_2^*)$ and
$l_{\beta_3}(u_1^*)$, which goes from negative
($l'_{\beta_3}(u_2)=0$, $u_1^*$ is the global maximizer) to
positive ($l'_{\beta_3}(u_1)=0$, $u_2^*$ is the global maximizer)
as $\beta_2$ goes from $m(b(\beta_1))$ to $m(a(\beta_1))$. Thus
there must be a unique $\beta_2:
m(b(\beta_1))<\beta_2<m(a(\beta_1))$ such that $u_1^*$ and $u_2^*$
are both global maximizers, and we denote this $\beta_2$ by
$r_{\beta_3}(\beta_1)$ (see Figure \ref{rcurve}). The parameter
values of $(\beta_1, r_{\beta_3}(\beta_1))$ are exactly the ones
for which positive and negative areas bounded by $l'_{\beta_3}(u)$
equal each other. An increase in $\beta_1$ induces an upward shift
of $l'_{\beta_3}(u)$, and must be balanced by a decrease in
$\beta_2$. Similarly, a decrease in $\beta_1$ induces a downward
shift of $l'_{\beta_3}(u)$, and must be balanced by an increase in
$\beta_2$. This justifies that $r_{\beta_3}$ is monotonically
decreasing in $\beta_1$.
\end{proof}

\begin{figure}
\centering
\includegraphics[clip=true, width=6in]{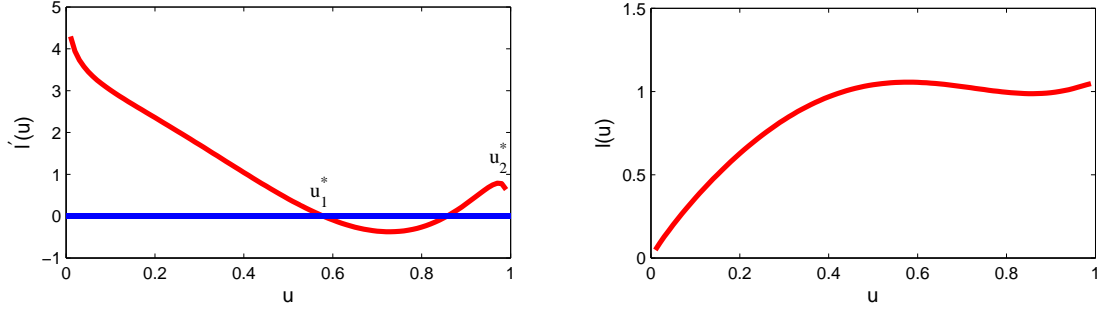}
\caption{Along the phase transition curve $r(\beta_1)$,
$l_{\beta_3}(u)$ has two local maximizers $u_1^*$ and $u_2^*$, and
both are global maximizers for $l_{\beta_3}(u)$. Graph drawn for
$\beta_1=2$, $\beta_2=-2.95$, $\beta_3=2$, $p=3$, and $q=5$.}
\label{rcurve}
\end{figure}

The following universality result shows that independent of the
specific local features that are incorporated into the exponential
random graph model (\ref{assumption}), the transition surface $S$
asymptotically approaches a common plane
$\beta_1+\beta_2+\beta_3=0$.

\begin{corollary}[Universality]
\label{coro} Fix $\beta_3\geq 0$. The transition curve
$\beta_2=r_{\beta_3}(\beta_1)$ displays a universal asymptotic
behavior as $\beta_1 \to -\infty$:
\begin{equation}
\lim_{\beta_1 \to
-\infty}|r_{\beta_3}(\beta_1)+\beta_1+\beta_3|=0.
\end{equation}
\end{corollary}

\begin{proof}
By Proposition \ref{max}, it suffices to show that as $\beta_1\to
-\infty$, $l_{\beta_3}(u; \beta_1, -\beta_1-\beta_3)$ has two
global maximizers $u_1^*$ and $u_2^*$. This is easy when we
realize that as $\beta_1\to -\infty$, $l_{\beta_3}(u; \beta_1,
-\beta_1-\beta_3)\to -\infty$ for every $u$ in $(0, 1)$. The
limiting maximizers on $[0, 1]$ are thus $u_1^*=0$ and $u_2^*=1$,
with $l_{\beta_3}(u_1^*)=l_{\beta_3}(u_2^*)=0$.
\end{proof}

\begin{proposition}
\label{char} As $\beta_3\geq 0$ varies, the transition curves
$\beta_2=r_{\beta_3}(\beta_1)$ (subject to $\beta_2 \geq 0$) trace
out a continuous surface $S$ with three bounding curves $C_1$,
$C_2$, and $C_3$.
\end{proposition}

\begin{proof}
The continuity of the transition surface $S$ follows easily once
we realize that it consists exactly of parameter values of
$(\beta_1, \beta_2, \beta_3)$ for which $l_{\beta_3}(u)$
(continuous in $\beta_1$, $\beta_2$, and $\beta_3$) has two global
maximizers. By Corollary \ref{coro}, $S$ displays a universal
asymptotic behavior: As $\beta_1 \to -\infty$, $\beta_2 \to
\infty$, and $\beta_3 \to \infty$, the distance between $S$ and
the plane $\beta_1+\beta_2+\beta_3=0$ shrinks to zero. Due to the
non-negativity constraints on $\beta_2$ and $\beta_3$, $S$ is
bounded by three curves $C_1$, $C_2$, and $C_3$: The curve $C_1$
is the intersection of $S$ with the ($\beta_1$, $\beta_2$) plane,
and is given by $\beta_2=r_0(\beta_1)$ (cf. Proposition
\ref{max}); The curve $C_2$ is the intersection of $S$ with the
($\beta_1$, $\beta_3$) plane, and is given analogously (with $p$
and $q$ switched in (\ref{l})); The curve $C_3$ is a critical
curve, and is traced out by the critical points $(\beta_1^c,
\beta_2^c)$ (\ref{point}) (subject to $\beta_2^c \geq 0$).
\end{proof}

\section{Critical Behavior}
\label{CB} By Propositions \ref{max} and \ref{char}, the
maximization problem (\ref{l}) is solved at a unique value $u^*$
off $S$, and at two values $u_1^*$ and $u_2^*$ on $S$ (the jump
from $u_1^*$ to $u_2^*$ is quite noticeable even for small
parameter values of $\beta$). Thus by Theorems \ref{CD} and
\ref{gen}, in the large $n$ limit, a typical $G_n$ drawn from
(\ref{pmf}) is indistinguishable from the Erd\H{o}s-R\'{e}nyi
graph $G(n, u^*)$ off the transition surface $S$, however, on the
transition surface $S$, the structure of $G_n$ is not completely
deterministic: It may behave like an Erd\H{o}s-R\'{e}nyi graph
$G(n, u_1^*)$, or it may behave like an Erd\H{o}s-R\'{e}nyi graph
$G(n, u_2^*)$. Since the limiting free energy density
$\psi_\infty^\beta$ encodes important information about the local
features of the random graph $G_n$ (see for example (\ref{E}) and
(\ref{Cov})), a thorough study of its analyticity properties is
fundamental to understanding the global structure of the
exponential model. The following theorems \ref{A} and \ref{B} are
dedicated to this goal. Together they complete the proof of our
main theorem (Theorem \ref{One}).

\begin{theorem}
\label{A} Consider a $3$-parameter exponential random graph model
(\ref{assumption}). The limiting free energy density
$\displaystyle \psi_\infty^\beta$ is not an analytic function on
the transition surface $S$.
\end{theorem}

\begin{proof}
Due to the jump between the two solutions $u_1^*$ and $u_2^*$ of
the maximization problem (\ref{l}), all the first derivatives
$\displaystyle \frac{\partial}{\partial
\beta_1}\psi_\infty^\beta$, $\displaystyle
\frac{\partial}{\partial \beta_2}\psi_\infty^\beta$, and
$\displaystyle \frac{\partial}{\partial \beta_3}\psi_\infty^\beta$
have (jump) discontinuities across the transition surface $S$,
except along the critical curve $C_3$:
\begin{equation}
\lim_{n\to \infty}\E^\beta t(H_1, G_n)=\lim_{n\to \infty} \E^\beta
t(H_1, G(n, u^*))=u^*= \frac{\partial}{\partial
\beta_1}\psi_{\infty}^\beta,
\end{equation}
\begin{equation}
\lim_{n\to \infty}\E^\beta t(H_2, G_n)=\lim_{n\to \infty} \E^\beta
t(H_2, G(n, u^*))=(u^*)^p= \frac{\partial}{\partial
\beta_2}\psi_{\infty}^\beta,
\end{equation}
\begin{equation}
\lim_{n\to \infty}\E^\beta t(H_3, G_n)=\lim_{n\to \infty} \E^\beta
t(H_3, G(n, u^*))=(u^*)^q= \frac{\partial}{\partial
\beta_3}\psi_{\infty}^\beta.
\end{equation}

To see that the transition across $C_3$ is second-order, we check
the first and second derivatives of $\psi_\infty^\beta$ in the
neighborhood of this curve. By Proposition \ref{max}, for every
$(\beta_1^c, \beta_2^c, \beta_3)$ on $C_3$, $l'_{\beta_3}(u;
\beta_1^c, \beta_2^c)$ is monotonically decreasing on $[0,1]$, and
the unique zero is achieved at $u_0$ (\ref{beta3}). Take any
$0<\epsilon<\min\{u_0, 1-u_0\}$. Set $\displaystyle
\delta=\min\{l'_{\beta_3}(u_0-\epsilon),
-l'_{\beta_3}(u_0+\epsilon)\}$. Consider $(\bar{\beta}_1,
\bar{\beta}_2, \bar{\beta}_3)$ so close to $(\beta_1^c, \beta_2^c,
\beta_3)$ such that $\displaystyle
|\bar{\beta}_1-\beta_1^c|+p|\bar{\beta}_2-\beta_2^c|+q|\bar{\beta}_3-\beta_3|<\delta$.
For every $u$ in $[0,1]$, we then have $\displaystyle
\left|l'_{\bar{\beta}_3}(u; \bar{\beta}_1,
\bar{\beta}_2)-l'_{\beta_3}(u; \beta_1^c,
\beta_2^c)\right|<\delta$. It follows that the zero
$u^*(\bar{\beta}_1, \bar{\beta}_2, \bar{\beta}_3)$ (or $u_1^*$ and
$u_2^*$) must satisfy $\left|u^*-u_0\right|<\epsilon$, which
easily implies the continuity of $\displaystyle
\frac{\partial}{\partial \beta_1}\psi_\infty^\beta$,
$\displaystyle \frac{\partial}{\partial
\beta_2}\psi_\infty^\beta$, and $\displaystyle
\frac{\partial}{\partial \beta_3}\psi_\infty^\beta$ at
$(\beta_1^c, \beta_2^c, \beta_3)$. Concerning the divergence of
the second derivatives, we compute
\begin{equation}
\frac{\partial^2}{\partial
\beta_1^2}\psi_{\infty}^\beta=\frac{\partial}{\partial \beta_1}u^*
=-\frac{1}{l''_{\beta_3}(u^*)}, \hspace{1cm}
\frac{\partial^2}{\partial
\beta_2^2}\psi_{\infty}^\beta=\frac{\partial}{\partial
\beta_2}(u^*)^p =-\frac{p^2(u^*)^{2p-2}}{l''_{\beta_3}(u^*)},
\end{equation}
\begin{equation}
\frac{\partial^2}{\partial
\beta_3^2}\psi_{\infty}^\beta=\frac{\partial}{\partial
\beta_3}(u^*)^q =-\frac{q^2(u^*)^{2q-2}}{l''_{\beta_3}(u^*)},
\hspace{1cm} \frac{\partial^2}{\partial \beta_1
\partial \beta_2}\psi_{\infty}^\beta=\frac{\partial}{\partial
\beta_1}(u^*)^p =-\frac{p(u^*)^{p-1}}{l''_{\beta_3}(u^*)},
\end{equation}
\begin{equation}
\frac{\partial^2}{\partial \beta_1
\partial \beta_3}\psi_{\infty}^\beta=\frac{\partial}{\partial
\beta_1}(u^*)^q =-\frac{q(u^*)^{q-1}}{l''_{\beta_3}(u^*)},
\hspace{1cm} \frac{\partial^2}{\partial \beta_2
\partial \beta_3}\psi_{\infty}^\beta=\frac{\partial}{\partial
\beta_2}(u^*)^q=-\frac{pq(u^*)^{p+q-2}}{l''_{\beta_3}(u^*)}.
\end{equation}
But as was explained in Proposition \ref{max}, as $(\bar{\beta}_1,
\bar{\beta}_2, \bar{\beta}_3)$ approaches $C_3$,
$l''_{\bar{\beta}_3}(u^*; \bar{\beta}_1, \bar{\beta}_2)$ converges
to zero. The desired singularity is thus justified.
\end{proof}

Real and complex analyticity are both defined in terms of
convergent power series. To examine the analyticity of the
limiting free energy density $\psi_\infty^\beta$ off the
transition surface $S$, we resort to an analytic implicit function
theorem, which may be interpreted in either the real or the
complex setting.

\begin{theorem}[Krantz-Parks \cite{KP}]
\label{KP} Suppose that the power series
\begin{equation}
F(x, y)=\sum_{\alpha, k}a_{\alpha, k}{x}^{\alpha}y^k
\end{equation}
is absolutely convergent for $|x|\leq R_1$ and $|y|\leq R_2$. If
$a_{0,0}=0$ and $a_{0,1}\neq 0$, then there exist $r_0>0$ and a
power series
\begin{equation}
\label{f} f(x)=\sum_{|\alpha|>0}c_{\alpha}{x}^{\alpha}
\end{equation}
such that (\ref{f}) is absolutely convergent for $|x|\leq r_0$ and
$F(x, f(x))=0$.
\end{theorem}

\begin{theorem}
\label{B} Consider a $3$-parameter exponential random graph model
(\ref{assumption}). Suppose $\beta_2$ and $\beta_3$ are
non-negative. Then the limiting free energy density $\displaystyle
\psi_\infty^\beta$ is an analytic function off the transition
surface $S$.
\end{theorem}

\begin{proof}
It is clear that $l_{\beta_3}(u; \beta_1, \beta_2)$ is analytic
for $u\in (0, 1)$, $\beta_1\in (-\infty, \infty)$, $\beta_2\in
(-\infty, \infty)$, and $\beta_3\in (-\infty, \infty)$. We show
that the maximizer $u^*$ for $l_{\beta_3}(u; \beta_1, \beta_2)$ is
an analytic function of $\beta$ off the transition surface $S$.
Fix $(\beta_1, \beta_2, \beta_3)$ not on $S$. For $(\bar{\beta}_1,
\bar{\beta}_2, \bar{\beta}_3)$ close to $(\beta_1, \beta_2,
\beta_3)$, we transform the function $l'_{\bar{\beta}_3}(u;
\bar{\beta}_1, \bar{\beta}_2)$ into a function $F(x, y)$ by
setting $x=(\bar{\beta}_1-\beta_1, \bar{\beta}_2-\beta_2,
\bar{\beta}_3-\beta_3)$ and $y=u-u^*(\beta_1, \beta_2, \beta_3)$.
It is easy to check that $F(x, y)$ satisfies all the conditions of
Theorem \ref{KP}: It has the desired domain of analyticity, is
locally absolutely convergent, and its first two coefficients are
given by
\begin{equation}
a_{0,0}=F(0, 0)=l'_{\beta_3}(u^*(\beta_1, \beta_2, \beta_3);
\beta_1, \beta_2)=0,
\end{equation}
\begin{equation}
a_{0,1}=\frac{\partial F}{\partial y}(0,
0)=l''_{\beta_3}(u^*(\beta_1, \beta_2, \beta_3); \beta_1,
\beta_2)\neq 0.
\end{equation}
The absolute convergence for $f(x)=u^*(\bar{\beta}_1,
\bar{\beta}_2, \bar{\beta}_3)-u^*(\beta_1, \beta_2, \beta_3)$ then
follows easily, which implies the analyticity of $u^*$ as a
function of $\beta$. As the composition of analytic functions is
analytic as long as the domain and range match up, by Theorem
\ref{CD}, this further implies the analyticity of
$\psi_\infty^\beta=l_{\beta_3}(u^*; \beta_1, \beta_2)$ as a
function of $\beta$ off the transition surface $S$, where the
maximizer $u^*$ is uniquely defined.
\end{proof}

\section*{acknowledgements} The author gratefully acknowledges the
support of the National Science Foundation through two
international travel grants, which enabled her to attend the 8th
World Congress on Probability and Statistics and the 17th
International Congress on Mathematical Physics, where she had the
opportunity to discuss this work. She is also thankful to the
anonymous referees for their useful comments and suggestions.

\end{document}